\documentclass[onecolumn,10pt]{article}
\usepackage[top=2.54cm, bottom=2.54cm, left=2.54cm, right=2.54cm,a4paper]{geometry}

\usepackage[utf8]{inputenc} 
\usepackage{textcomp} 

\usepackage{amsmath,amsthm}
\usepackage{bm}  
\usepackage{bbm}

\usepackage[backend=biber,sorting=none,style=phys,biblabel=brackets,backref=false,bibencoding=utf8,maxnames=20,eprint=true]{biblatex}
\makeatletter
\pretocmd{\blx@head@bibintoc}{\phantomsection}{}{\ddt}
\makeatother
\DefineBibliographyStrings{english}{%
  backrefpage = {page},
  backrefpages = {pages},
}
\usepackage{titlesec}
\usepackage{caption}
\titleformat*{\section}{\bfseries}
\titleformat*{\subsection}{\normalsize\bfseries}
\titleformat*{\subsubsection}{\bfseries}
\titleformat*{\paragraph}{\bfseries}
\titleformat*{\subparagraph}{\bfseries}

\titlespacing\section{0pt}{12pt plus 4pt minus 2pt}{2pt plus 2pt minus 2pt}

\usepackage{setspace}

\usepackage{xcolor}
\definecolor{dullmagenta}{rgb}{0.4,0,0.4}   
\definecolor{darkblue}{rgb}{0,0,0.4}
\usepackage[bookmarks,colorlinks,breaklinks]{hyperref}  
\hypersetup{linkcolor=red,citecolor=blue,filecolor=dullmagenta,urlcolor=darkblue} 

\RequirePackage[charter, greekuppercase=italicized,cal=cmcal]{mathdesign}
\usepackage[T1]{fontenc}

\usepackage{verbatim}
\usepackage{xcolor}

\usepackage{braket}
\newcommand{\ketbra}[1]{|#1\rangle\langle #1|}

\renewcommand{\epsilon}{\varepsilon}
\renewcommand{\phi}{\varphi}
\def\eps{\epsilon}
\def\tr{{\rm Tr}}

\def\eps{\epsilon}
\def\id{\mathbbm{1}}

\usepackage{nicefrac}

\def\bsc{\text{BSC}}
\def\pure{\text{BPC}}
\def\bec{\text{BEC}}
\def\cE{\mathcal E}
\def\cF{\mathcal F}

\def\pnr{\textup{P}_{0}}
\def\pinfo{\textup{P}_{1}}
\def\p{\textup{P}}

\usepackage{xspace}

\def\cP{\mathcal P}
\def\cN{\mathcal N}
\def\cW{\mathcal W}

\usepackage{mathtools}

\definecolor{mblue}{rgb}{0.368417, 0.506779, 0.709798}
\definecolor{morange}{rgb}{0.880722, 0.611041, 0.142051}
\definecolor{mgreen}{rgb}{0.560181, 0.691569, 0.194885}
\definecolor{mred}{rgb}{0.922526, 0.385626, 0.209179}
\definecolor{mpurple}{rgb}{0.528488, 0.470624, 0.701351}

\newtheorem{theorem}{Theorem}

\newtheorem{lemma}{Lemma}
\newtheorem{corollary}{Corollary}

\newtheorem{definition}{Definition}
\newtheorem{fact}{Fact}

\usepackage{titling}

\pretitle{\begin{center}\large \bfseries}
\posttitle{\par\end{center}}
\predate{\vspace{-1\baselineskip}\begin{center}\small }
\postdate{\end{center}}
\newcommand*{\cC}{\mathcal{C}}
\newcommand*{\cD}{\mathcal{D}}

\newcommand{\proj}[1]{|#1\rangle\!\langle #1|}

\usepackage{braket}

\newcommand{\dyad}[2]{|#1\rangle\langle #2|}

\DeclareFontFamily{U}{mathx}{\hyphenchar\font45}
\DeclareFontShape{U}{mathx}{m}{n}{<-> mathx10}{}
\DeclareSymbolFont{mathx}{U}{mathx}{m}{n}
\DeclareMathAccent{\widebar}{0}{mathx}{"73}

\def\cP{\mathcal P}
\def\cN{\mathcal N}
\def\cW{\mathcal W}
\def\pnr{\textup{P}_{0}}
\def\pinfo{\textup{P}_{1}}
\def\p{\textup{P}}

\def\bsc{\text{BSC}}
\def\pure{\text{BPC}}
\def\bec{\text{BEC}}

\usepackage{pgfplots}
\pgfplotsset{compat=1.14}
\usetikzlibrary{spy,fit}

\newcommand{\syndrome}{S_{2}^{n}}
\newcommand{\codeblock}{V_1S_{2}^{n}}
\newcommand{\decoder}{\mathcal{U}_{Y^n\to V_1S_{2}^{n}}^{\mathcal{D}}}
\newcommand{\strings}{w_{2}^{n}\in \{0,1\}^{n-1}}

\newcommand{\identity}{\id}

\allowdisplaybreaks    

\begin{document}

\title{Non-additivity in classical-quantum wiretap channels}

\author{\normalsize  Arkin Tikku$^1$, \href{https://orcid.org/0000-0002-0428-3429}{\color{black} Mario Berta}$^2$, and \href{http://orcid.org/0000-0003-2302-8025}{\color{black} Joseph M.\ Renes}$^3$\\
\small ${}^1$Centre for Engineered Quantum Systems, School of Physics, University of Sydney, Australia\\
\small ${}^2$Department of Computing, Imperial College London, England\\
\small ${}^3$Institute for Theoretical Physics, ETH Z\"urich, Switzerland}



\date{ \today}

\maketitle

\vspace{-2em}
\begin{abstract}
Due to Csisz\'ar and K\"orner, the private capacity of classical wiretap channels has a single-letter characterization in terms of the private information. For quantum wiretap channels, however, it is known that regularization of the private information is necessary to reach the capacity. Here, we study hybrid classical-quantum wiretap channels in order to resolve to what extent quantum effects are needed to witness non-additivity phenomena in quantum Shannon theory. For wiretap channels with quantum inputs but classical outputs, we prove that the characterization of the capacity in terms of the private information stays single-letter. Hence, entangled input states are of no asymptotic advantage in this setting. For wiretap channels with classical inputs, we show by means of explicit examples that the private information already becomes non-additive when either one of the two receivers becomes quantum (with the other receiver staying classical). This gives non-additivity examples that are not caused by entanglement and illustrates that quantum adversaries are strictly different from classical adversaries in the wiretap model.
\end{abstract}


\section{Introduction}\label{sec:intro}

In contrast to its classical counterpart, non-additivity phenomena of entropic expressions already make an appearance in some basic settings of quantum Shannon theory. This includes the quantum capacity \cite{shor_quantum_1996,divincenzo_quantum-channel_1998}, the private capacity \cite{smith_structured_2008,li_private_2009}, and the classical capacity of quantum channels \cite{hastings_superadditivity_2009}.
One of the most vexing such problems was posed by the additivity conjecture for the Holevo information \cite{shor_equivalence_2004}. 
Hastings disproved the conjecture with an example in which using entangled inputs to a quantum channel boosts the rate at which information can be transmitted~\cite{hastings_superadditivity_2009}. 

Here, we investigate private communication over wiretap channels to understand the essential quantum properties needed for entropic channel capacity formulas to become non-additive. Originally introduced by Wyner~\cite{wyner_wire-tap_1975}, a wiretap channel $\cW_{BC|A}$ has one input system $A$ for Alice and two outputs $B$ and $C$ to Bob and Charlie, respectively. The goal is then to transmit classical information represented by some finite message set $M=\{1,...,m\}$ from Alice to Bob, using a channel coding scheme of block size $n$ with encoder $\cE$ and decoder $\cD$ such that Bob can make a reliable inference $\hat{M}$ about the transmitted message $M$, without any information leaking to Charlie.
More precisely, the probability of error $P\{\hat{m}\neq m\}$ should be small for all $m$, say $P\{\hat{m}\neq m\}\leq \eps$ for some $\eps\in (0,1)$.
Moreover, Charlie's outputs for each $m$ should be nearly indistinguishable. Taking the general case of quantum state $\rho_C^m$ for input $m$, we require $\tfrac12\|\rho_C^m-\sigma_C\|_1\leq \eps$ for some $\sigma_C$ and all $m\in M$, where for simplicity we take the same parameter $\eps$ (see, e.g.\ \cite[Chapter 23]{wilde_quantum_2017}). 
A given coding scheme of size $n$ and parameter epsilon has a rate $R(n,\eps)=\frac1n\log |M|$, and the supremum of achievable rates for integer $n$ and $\eps\in(0,1)$ defines the capacity $\p(\cW)$.

In this work, we are interested in whether the private capacity $\p(\cW)$ has a single-letter expression in terms of the private information
\begin{align}\label{eq:private-info}
\pinfo(\cW):=\max_{p_{V},\rho_{A}^v} I(V:B)_\omega-I(V:C)_\omega\\
\mathrm{with}\quad\omega_{VBC}:=\sum_v p_{V}(v)\ketbra{v}_V\otimes \cW_{BC|A}\left(\rho_{A}^v\right)\,.
\end{align}
Here, $p_V$ is a probability distribution over an auxiliary random variable $V$ and $\rho_{A}^v$ is a quantum state conditional on the value of $V$ and with support on $A$.
For quantum wiretap channels, the private capacity is known to be characterized by the regularization of the private information~\cite{devetak_private_2005,cai_quantum_2004}, i.e.\
\begin{align} \label{eqn:priv_capacity}
\p(\cW)=\lim_{n\to \infty} \frac{1}{n} \pinfo(\cW^{\otimes n})\,.
\end{align}
This is precisely how the analysis proceeds in the case of classical wiretap channels as well, but Csisz\'ar and K\"orner~\cite{csiszar_broadcast_1978} further showed that the private information is additive. That is, we have $\pinfo(\cW_1\otimes \cW_2)=\pinfo(\cW_1)+\pinfo(\cW_2)$ making the regularization unnecessary. Our main result in this paper is that if two of the parties in the wiretap channel are chosen to be classical, then $\p(\cW)=\pinfo(\cW)$ when the \emph{input} is quantum, while there exist channels for which $\p(\cW)>\pinfo(\cW)$ if either \emph{output} is quantum. Hence, in this scenario neither do entangled inputs allow for non-additivity effects to occur, nor are they necessary for them!

The rest of the paper is structured as follows. In Section~\ref{sec:setting} we fix our notation, and in Section~\ref{sec:privateinfo} we examine general properties of wiretap channels and the private information. In particular, we show that one can always take $|V|\leq |A|^2$ in the private information optimization \eqref{eq:private-info}. Then we show additivity for quantum inputs in Section~\ref{sec:quantum-alice} and give the non-additive examples in Section~\ref{sec:quantum-bob} and Section~\ref{sec:quantum-eve}. 


\section{Setting}\label{sec:setting}

\subsection{Systems}

Quantum systems are denoted by $A,B,C$ and have finite dimensions $|A|,|B|,|C|$, respectively. Quantum states are linear, positive semi-definite operators of trace one and denoted by $\rho_A\in\cD(A)$, where the subscript denotes the support of the operator. Quantum states $\rho_A\in\cD(A)$ are called pure if they are of rank one, in which case we also write $\rho_A=\proj{\psi}_A$. Quantum channels $\cW_{B|A}$ from $A$ to $B$ correspond to completely positive and trace-preserving maps from the linear operators on $A$ to the linear operators on $B$.
Classical systems are denoted by $V,W,X,Y,Z$ and have finite dimensions $|V|,|W|,|X|,|Y|,|Z|$, respectively. Classical states are density matrices diagonal in the computational basis $\{\proj{x}\}_{x\in X}$ and denoted by $\rho_X \in \cD(X)$. Classical channels from $X$ to $Y$ correspond to conditional probability distributions $p_{Y|X}(y|x)$, but may also at times be denoted by $\cN_{Y|X}$, with the support indicating the classical domain and target. The notation $X^m=(X_1,...,X_{m})$ denotes an $m$-tuple of registers and will be used in the context of channel coding to denote a code-block of length $m$ that encodes a logical system $X$.


\subsection{Entropies}

For $\rho_{ABC}\in\cD(ABC)$ and its reduced states, the {\it entropy} is defined as $H(A)_\rho:=-\mathrm{Tr}\left[\rho_A\log\rho_A\right]$ (where logarithms are taken base 2), the {\it conditional entropy} of $A$ given $B$ as $H(A|B)_\rho:=H(AB)_\rho-H(B)_\rho$, the {\it mutual information} between $A$ and $B$ as $I(A:B)_\rho:=H(A)_\rho+H(B)_\rho-H(AB)_\rho$, and the {\it conditional mutual information} between $A$ and $B$ given $C$ as $I(A:B|C)_\rho:=H(AC)_\rho+H(BC)_\rho-H(ABC)_\rho-H(C)_\rho$. Here and henceforth any quantum definition applies to classical probability distributions as well\,---\,by embedding them as matrices diagonal in the computational basis $\{\proj{x}\}_{x\in X}$.


\subsection{Wiretap Channels}

A {\it wiretap channel} is given by a quantum channel $\cW_{BC|A}$ with one sender Alice $A$ and two receivers Bob $B$ and Charlie $C$, where Bob acts as the legitimate receiver and Charlie as the adversarial party. Note that we do not require $\cW_{BC|A}$ to be an isometric channel (as often done in the literature), i.e.\ the channel to Charlie is not necessarily the complement of the channel to Bob. We are then interested in hybrid classical-quantum settings, where some of the systems are classical. This leads us to use the following definition:

\begin{definition}
A reduction $\cW_{B|A}$ $(\cW_{C|A})$ of a wiretap channel $\cW_{BC|A}$ to the legitimate receiver (adversarial party) is obtained by tracing out the adversarial party (legitimate receiver), i.e 
\begin{align}
    \cW_{B|A}(\cdot):=\tr_{C}\left[\cW_{BC|A}(\cdot)\right] \quad \textit{or} \quad \cW_{C|A}(\cdot):=\tr_{B}\left[\cW_{BC|A}(\cdot)\right]\,.
\end{align}
We will denote wiretap channels with classical inputs as $\cW_{BC|X}:=\cW_{B|X}/\cW_{C|X}$, indicating their construction from their reductions via $\cW_{B|X}/\cW_{C|X}:= \left(\cW_{B|X'}\otimes\cW_{C|X''}\right)\circ \cC_{X'X''|X}$, where $\cC_{X'X''|X}$ is a stochastic map that creates a copy of the classical input system $X$. 
\end{definition}

\noindent We will consider reductions that are either fully classical channels or classical-quantum channels. 
From these, we will then construct the following types of wiretap channels:

\begin{definition}
A quantum-classical-classical (qcc) wiretap channel with quantum sender Alice $A$, but classical receivers Bob $Y$ and Charlie $Z$ is given by
\begin{align}
\cW_{YZ|A}(\cdot)=\sum_{y,z}\mathrm{Tr}\big[\Lambda_A^{y,z}(\cdot)\big]\proj{y}_Y\otimes\proj{z}_Z
\end{align}
with $\{\Lambda_A^{y,z}\}_{y,z}$ forming a POVM, such that $\Lambda_A^{y,z}\geq0$ and $\sum_{y,z}\Lambda_A^{y,z}=\identity_A$.
\end{definition}

This setting is notable because it allows for entangled inputs to the wiretap channel, but only separable states at the outputs. It thus raises the question of whether using entangled input states can boost the rate at which information can be transmitted, despite the fact that each reduction is an entanglement-breaking channel. 

\begin{definition}
A classical-quantum-classical (cqc) wiretap channel with classical sender Alice $X$ and classical adversarial receiver Charlie $Z$, but quantum legitimate receiver Bob $B$ is given by
\begin{align}
\cW_{BZ|X}(\cdot)=\sum_x\bra{x}\cdot\ket{x}\left(\sum_{z}p(z|x)\rho_B^z\otimes\proj{z}_Z\right)
\end{align}
for conditional probability distributions $p(z|x)$ and quantum states $\rho_B^z$.
\end{definition}

Here, only separable states that are diagonal in the computational basis are allowed as inputs to the channel, so one might think that the private information is additive in this case. 
\begin{definition}
A classical-classical-quantum (ccq) wiretap channel with classical sender Alice $X$ and classical legitimate receiver Bob $Y$, but quantum adversarial receiver Charlie $C$ is given by
\begin{align}
\cW_{YC|X}(\cdot)=\sum_x\bra{x}\cdot\ket{x}\left(\sum_yp(y|x)\proj{y}_Y\otimes\rho_{C}^y\right)
\end{align}
for conditional probability distributions $p(y|x)$ and quantum states $\rho_C^y$.
\end{definition}

Again, one might expect the private information to be additive given that only separable states diagonal in the computational basis are allowed for inputs. Some simple channels that we will employ to construct wiretap channel examples are as follows.
First, the binary symmetric channel with crossover probability $p$ is denoted $\bsc{(p)}$.
It may also be thought of as a quantum channel $\cN_{B|A}(\cdot):=(1-p)(\cdot)+ pX(\cdot)X$ for the Pauli $X$-matrix on system $B$.
The binary erasure channel with erasure probability $p$ is denoted $\bec(p)$. 
Again a fully quantum binary erasure channel $\cN_{B|A}$ may be defined by $\cN_{B|A}(\cdot)= (1-p)(\cdot)+p \ketbra{e}_{B}$, where $\ket{e}_{B}$ is orthogonal to $A$.
\begin{definition}
The binary pure state channel $\pure{(f)}$ with fidelity $f$ is a classical-quantum channel $\cN_{A|X}$ with pure state outputs $\dyad{\psi}{\psi}$ and $\dyad{\phi}{\phi}$, such that $f:=|\langle\phi|\psi\rangle|$ and
\begin{align}
 \dyad{0}{0}_{X}\longrightarrow \dyad{\psi}{\psi}_{A} \quad \text{and} \quad \dyad{1}{1}_{X}\longrightarrow \dyad{\phi}{\phi}_{A}\,.
\end{align}
\end{definition}



\section{Properties of the private information}\label{sec:privateinfo}

To evaluate the private information for the case of classical inputs in later sections, we follow ideas from the classical work \cite{ozel_wiretap_2013} and start by rewriting
\begin{align}\label{eq:p1asf}
\pinfo(\cW)=\max_{\left\{p_V,\rho_{A}^v \right\} } f_{\cW}\left(\mathbb E_{V} \left[\rho_{A}^v\right]\right)-\mathbb E_V \left[f_{\cW}\left(\rho_{A}^v\right)\right]\quad\mathrm{with}\quad f_{\cW}(\rho_A):=I(X:B)_\omega-I(X:C)_\omega\,,
\end{align}
where $\rho_A=\sum_{x}p_{X}(x)\ketbra{x}_{A}$ for $\{\ket{x}_{A}\}$ the eigenbasis of $\rho_A$, and $\omega_{XBC}=\sum_{x}p_{X}(x)\ketbra{x}_{X}\otimes \cW_{BC|A}\left(\ketbra{x}_{A}\right)$. This expression still holds for quantum inputs and is in general hard to evaluate, as the underlying optimization problem is non-convex. However, we immediately have the upper bound
\begin{align} \label{eq:upper_bound}
\pinfo(\cW)\leq\max_{\rho_A} f_{\cW}(\rho_A)-\min_{\rho_A} f_{\cW}(\rho_A)\,.
\end{align}
since the minimum of any function is a lower bound to its expectation value.  
Moreover, whenever $f_{\cW}\geq0$ for all input states $\rho_A$, then $\min_\rho f_{\cW}=0$, with any pure input state $\rho_A=\ketbra{x}_{A}$ being a valid minimizing argument and hence
\begin{align}\label{eq:more-capable}
\pinfo(\cW)=\pnr(\cW):=\max_{\rho_A} f_{\cW}(\rho_A)\,.
\end{align}
We call wiretap channels $\cW$ with $f_{\cW}\geq 0$ {\it more-capable} \cite{korner_comparison_1977}, where for the classical case this can be seen as the sufficiency to choose $V=X$ in \eqref{eq:private-info}. Whenever the function $f_{\cW}$ is concave, we call the wiretap channel $\cW$ {\it less-noisy}. Such channels are then in particular also more-capable and thus the optimization problem in \eqref{eq:more-capable} becomes convex and therefore easily tractable. Contrary to the classical case, however, we do  not know if less-noisy wiretap channels have additive private information in general.\footnote{Our definition of less-noisy and more-capable conflicts with Watanabe's work for the quantum case \cite{watanabe_private_2012}. He introduces a notion that we might be tempted to call completely less-noisy and completely more-capable, as it involves the less noisy or more capable condition applied to arbitrarily many instances of the channel. His notions are equivalent to ours for the classical case \cite{korner_comparison_1977}, and lead to an additive private information.} We also have the notions of {\it anti-less-noisy} and {\it anti-more-capable} wiretap channels, where the roles of the legitimate receiver Bob and adversarial receiver Charlie are interchanged. Note that for anti-less-noisy channels $\pinfo(\cW)=0$, as well as that for anti-more-capable channels $\pnr(\cW)=0$. Finally, a well-known sufficient criterion for additivity of the private information for general (i.e. not necessarily isometric) quantum wiretap channels, is {\it degradability} \cite{cover_broadcast_1972,devetak_capacity_2005,smith_private_2008}. That is, when there exists a channel $\cE_{C|B}$ such that $W_{C|A}=\cE_{C|B}\circ W_{B|A}$. The private information is also additive for {\it anti-degradable} channels, when there exists a channel $\cF_{B|C}$ such that $W_{B|A}=\cF_{B|C}\circ W_{C|A}$. However, then we immediately have $\p(\cW)=\pinfo(\cW)=0$.

Next, we prove a cardinality upper bound for the private information.

\begin{lemma}[Fenchel-Eggleston {\cite[Theorem 18]{eggleston_convexity_1958}}]\label{lem:Fenchel}
Let $S\subseteq\mathbb{R}^n$ such that $S=\bigcup_{i=1}^nS_i$ with $S_i$ connected. Then, we have for every $y\in\text{conv}(S)$ that there exists $S'\subseteq S$ with $|S'|\leq n$ such that $y\in\text{conv}(S')$.
\end{lemma}

\begin{lemma}[Cardinality bound]\label{lem:card_bound}
Let registers $V,A$ be defined as in \eqref{eq:private-info}. Then, the maximization of \eqref{eq:private-info} is achievable using an ensemble over input states for which $|V|\leq|A|^{2}$. For classical input systems $X$ the bound reduces to $|V|\leq|X|$.
\end{lemma}

\begin{proof}
Note that we can alternatively write \eqref{eq:p1asf} as
\begin{align}
\pinfo(\cW)=\max_{\left\{p_V,\rho_{A}^v\right\}} g_{\cW}(\mathbb E_{V} \left[\rho_{A}^v\right])-\mathbb E_V \left[g_{\cW}(\rho_{A}^v)\right]\,
\end{align}
with $g_{\cW}(\rho_A):=H(B)_{\cW(\rho)}-H(C)_{\cW(\rho)}$, since 
\begin{align}
g_{\cW}(\mathbb E_{V} \left[\rho_{A}^v\right])-\mathbb E_V \left[g_{\cW}(\rho_{A}^v)\right]
&= H\Big(\sum_{v}p_{V}(v)\rho_{B}^{v}\Big)-\sum_{v}p_{V}(v)H(\rho_{B}^{v})\notag\\
&\phantom{=}-H\Big(\sum_{v}p_{V}(v)\rho_{C}^{v}\Big)+\sum_{v}p_{V}(v)H(\rho_{C}^{v})\,.
\end{align}
Then, following \cite[Theorem 17.11]{csiszar_information_2011} we consider the function $h$ which maps any state $\rho$ to its  Bloch vector, as well as the value $g_\cW(\rho)$. Recall that the Bloch vector has $d^2-1$ components, where $d$ is the Hilbert space dimension, and completely specifies the state $\rho$. Since $h$ is continuous, the image $S$ of the set of states under $h$ is a compact, convex, and connected set in $\mathbb R^{d^2}$. Now, suppose $p_V$ and $\rho_{A}^{v}$ are optimal, leading to an average state $\rho_A:=\sum_{v}p_{V}(v)\rho_{A}^{v}$. By Fenchel-Eggleston's strengthening of Carath\'eodory's theorem, as given in Lemma~\ref{lem:Fenchel}, $h(\rho_A)$ can be represented as $\sum_{i=1}^{d^2} p_i s_i$ for suitable probabilities $p_i$ and points $s_i\in S$. For each $i$, the first $d^2-1$ components of $s_i$ specify a state $\rho_i$, while the last component is $h_\cW(\rho_i)$. Therefore, there exists a random variable $V'$ of cardinality $d^2$ with $p_{V'}(i)=p_i$ and a preparation map $\cP'_{A|V}$ with $\cP'_{A|V=i}=\rho_i$ which is also optimal. If we restrict the possible inputs to $\cW$ to form a commuting set, i.e.\ diagonal in some basis, then only $d-1$ components are needed for the Bloch vector, and we recover the classical cardinality bound $|V|\leq |X|$.
\end{proof}


\section{Quantum Sender Alice}\label{sec:quantum-alice}

Here, we prove that the private information is additive for wiretap channels $\cW_{YZ|A}$ with quantum input system $A$ but classical output systems $YZ$.

\begin{theorem} \label{thm:qAlice}
Let $\cW_{1}:=\cW_{Y_1Z_1|A_1}$ and $\cW_{2}:=\cW_{Y_2Z_2|A_2}$ be two qcc-wiretap channels with quantum senders $A_1,A_2$, classical legitimate receivers $Y_1,Y_2$ and classical adversaries $Z_1,Z_2$. Then, we have
\begin{align}
\pinfo(\cW_{1}\otimes \cW_{2})=\pinfo(\cW_{1})+\pinfo(\cW_{2})\,.
\end{align} 
\end{theorem}

\begin{proof}
The proof makes use of a variant of the classical key identity in \cite[Lemma 17.12]{csiszar_information_2011}. Suppose $\rho_{VA_1A_2}$ is the optimizer in $\pinfo(\cW_1\otimes \cW_2)$ and call the outputs $Y_1Y_2$ for Bob and $Z_1Z_2$ for Charlie. Let $\cW_1:=\cW_{Y_1Z_1|A_1}$ and $\cW_2:=\cW_{Y_2Z_2|A_2}$, then for the probability distribution $p_{VY_1Y_2Z_1Z_2}=\left(\cW_{Y_1Z_1|A_1}\otimes \cW_{Y_2Z_2|A_2}\right)(\rho_{VA_1A_2})$
we have
\begin{align}
\pinfo(\cW_1\otimes \cW_2)
&=I(V:Y_1Y_2)_{p}-I(V:Z_1Z_2)_{p}\\
&=I(V:Y_1|Z_2)_{p}-I(V:Z_1|Z_2)_{p}+I(V:Y_2|Y_1)_{p}-I(V:Z_2|Y_1)_{p}\\
&\leq \max_{z_2}\left[ I(V:Y_1|Z_2=z_2)_{p}-I(V:Z_1|Z_2=z_2)_{p}\right]\notag \\
&\phantom{\leq}+\max_{y_1}\left[I(V:Y_2|Y_1=y_1)_{p}-I(V:Z_2|Y_1=y_1)_{p}\right]\,.
\end{align}
The second equation follows using the chain rule for the conditional mutual information, while the first inequality follows since conditioning is equivalent to averaging. Now, consider the first maximization, for which we require the joint distribution $p_{VY_1Z_1Z_2}$. Suppose the optimal input state has the form
\begin{align}
\rho_{VA_1A_2}=\sum_v p_{V}(v)\ketbra{v}_V\otimes \varphi^{v}_{A_1A_2}\,.
\end{align}
The channels are just measurements, so letting $\Lambda^{y_1,z_1}$ and $\Gamma^{y_2,z_2}$ be the associated POVM elements for $\cW_1$ and $\cW_2$ respectively yields 
\begin{align}
p_{VY_1Z_1Z_2}(v,y_1,z_1,z_2)
=p_V(v)\sum_{y_2}\tr\left[(\Lambda_{A_1}^{y_1,z_1}\otimes \Gamma_{A_2}^{y_2,z_2})\varphi_{A_1A_2}^{v}\right]\,.
\end{align}
Now, define the normalized states $\sigma_{A_1}^{v,z_2}$ via
\begin{align}
p_{Z_2|V=v}(z_2)\sigma_{A_1}^{v,z_2}=\sum_{y_2}\tr_{A_2}\left[\Gamma_{A_2}^{y_2,z_2}\rho_{A_1A_2}^{v}\right]\,,
\end{align}
for $p_{Z_2|V}$ the conditional distribution computed from the distribution $p_{VY_1Z_1Z_2}$. Writing in the decomposition $p_V(v)\, p_{Z_2|V=v}(z_2)=p_{Z_2}(z_2)\,p_{V|Z_2=z_2}(v)$, we obtain
\begin{align}
p_{VY_1Z_1|Z_2=z_2}(v,y_1,z_1)=p_{V|Z_2=z_2}(v)\,\tr\left[\Lambda^{y_1,z_1}_{A_1}\sigma^{v,z_2}_{A_1}\right]\,.
\end{align}
Thus, we have confirmed that when conditioning on the value of $Z_2$, the outputs $Y_1$ and $Z_1$ are related to $V$ via $\cW_{1}$ composed with a preparation channel $\cP_{A_1|VZ_2}$. Therefore, we have 
\begin{align}
\max_{z_2} I(V:Y_1|Z_2=z_2)_p-I(V:Z_1|Z_2=z_2)_p\leq \pinfo(\cW_1)\,.
\end{align}
A similar argument holds for the second term, implying $\pinfo(\cW_1\otimes \cW_2)\leq \pinfo(\cW_1)+\pinfo(\cW_2)$. The other inequality holds since, for the optimal $\rho_{V_1A_1}$ ($\rho_{V_2A_2}$) in $\pinfo(\cW_1)$ ($\pinfo(\cW_2)$), the product state $\rho_{V_1A_1}\otimes \rho_{V_2A_2}$ is feasible in $\pinfo(\cW_1\otimes \cW_2)$ with $V=(V_1,V_2)$.
\end{proof}

Notably, this shows that entangled input states are of no use and gives a novel single-letter characterization in quantum Shannon theory. 
Applying Theorem \ref{thm:qAlice} inductively to 
\eqref{eqn:priv_capacity} gives the following characterization.

\begin{corollary}
For quantum-classical-classical wiretap channels $\cW_{YZ|A}$ we have that $\p(\cW_{YZ|A})=\pinfo(\cW_{YZ|A})$.
\end{corollary}


\section{Quantum legitimate receiver Bob}\label{sec:quantum-bob}

Now let us turn to the question of classical inputs. 
By means of an explicit counterexample, here we show that the private information is non-additive for wiretap channels $\cW_{BZ|X}$ with classical input system $X$ and classical adversary $Z$, but quantum legitimate receiver $B$. 
Hence, somewhat surprisingly, classical adversaries already make classical-quantum channel coding amenable to non-additivity effects in the wiretap model. Note that this is not the case for classical communication over classical-quantum channels, whose capacity has the same single letter maximization of mutual information as the capacity for the purely classical channel.

\begin{theorem}\label{thm:qBob}
There exists a cqc-wiretap channel $\cW_{BZ|X}$ with $\p\left(\cW_{BZ|X}\right)>\pinfo\left(\cW_{BZ|X}\right)$.
\end{theorem}
In fact, we can construct a whole family of counterexamples, based on the parameterized wiretap channel $\cW_{BZ|X}[r]$, where $\cW_{B|X}[r]=\pure(r)$ and $\cW_{Z|X}[r]=\bec((1-r)^2)$, as depicted in Fig.~\ref{fig:Q_Bob_chann}. 
We construct a simple block preprocessing of a single bit input to two channel inputs using a parity code to show that the private information of two uses of $\cW_{BZ|X}[r]$ is positive for some values of $r$ for which the private information of a single use is zero.

\begin{figure}[t]
    \centering
    \begin{tikzpicture}
    \tikzstyle{box}=[draw,minimum width=25mm,minimum height=10mm]
    \node[box] (bec) at (0,-.8) {$\text{BEC}((1{-}r)^2)$};
    \node[box] (bpc) at (0,.8) {$\text{BPC}(r)$};
    \node (y) at (3,.8) {$B$};
    \node (c) at (3,-.8) {$Z$};
    \draw (bpc) -- (y);
    \draw (bec) -- (c);
    \node[inner sep=0,outer sep=0] (ctrl) at (-2,0) {};
    \node (x) at (-4,0) {$X$};
    \draw (x) -- (-2,0);
    \draw (ctrl) |- (bec);
    \draw (-2,0) |- (bpc);
    \fill (ctrl) circle (.4mm);
    \node[fit=(ctrl) (bpc) (bec), draw,dashed,inner sep=10pt] (W) {};
    \node[below] at (W.south) {$\cW_{BZ|X}[r]$};
    \end{tikzpicture}
    \caption{The channel $\cW_{BZ|X}[r]$, composed of $\pure(r)$ to Bob and $\bec((1-r)^2)$ to Charlie.}
    \label{fig:Q_Bob_chann}
\end{figure}

\begin{proof}
First, we show that the private information $\pinfo{\left(\cW_{BZ|X}[r]\right)}$ vanishes for channel parameter values $r\geq \hat r$, where $\hat r\approx 0.5424$ satisfies 
\begin{align}\label{eq:quantumBobthreshold}
\frac{2\hat r^2(\hat r -2)}{\hat r^2-1}=\log\left(\frac{1+\hat r }{1-\hat r }\right)\,.
\end{align}
Vanishing private information is the statement that the channel is anti-less-noisy, which is the case when $f_{\cW[r]}:\rho_X\to I(X':B)_\omega-I(X':Z)_\omega$ from \eqref{eq:p1asf} is convex. 
Here $\omega_{X'BZ}=\sum_{x}p_{X}(x)\ketbra{x}_{X'}\otimes\cW_{BZ|X}[r](\ketbra{x}_X)$.
As $\cW_{BZ|X}[r]$ has binary input, the cardinality bound $|V|\leq2$ from Lemma \ref{lem:card_bound} considerably simplifies the analysis. 
In Appendix \ref{app:A} we calculate the second derivative of $f_{\cW[r]}$ and find that it becomes positive for $r\geq \hat r$. 

To see that the private capacity is not zero for some $r\geq \hat r$, consider the preprocessing map $\mathcal P$ based on the $n=2$ parity encoding, which maps 
\begin{subequations}
\label{eq:parityprep}
\begin{align}
&\dyad{0}{0}_{X}\longrightarrow \frac{\dyad{00}{00}_{X^2}+\dyad{11}{11}_{X^2}}{2}\quad\text{and}\\
&\dyad{1}{1}_{X}\longrightarrow \frac{\dyad{01}{01}_{X^2}+\dyad{10}{10}_{X^2}}{2}\,.    
\end{align}
\end{subequations}

Due to symmetry, the rate $\frac12\pinfo(\cW_{BZ|X}[r]^{\otimes 2}\circ \mathcal P_{X^2|X})$ of the combined preprocessing and wiretap channel achieves the upper bound \eqref{eq:upper_bound}.
The maximal input to $f$ is the uniform distribution, and the minimal inputs occur symmetrically at distributions $(1-q,q)$ and $(q,1-q)$, for the appropriate value of $q$.
This implies that the optimal noisy preprocessing has a uniform $V$ and $P_{X|V}$ given by BSC$(q)$. 
As a concrete example, taking $r=0.543$ and $q\approx 0.2281$ gives a rate of roughly $0.0003$. 
This is a miniscule value, but positive. 
\end{proof}

\begin{figure}[ht]
\begin{center}
\begin{tikzpicture}

\begin{axis}[
name=mainplot,
grid,
enlargelimits = false,
width=\textwidth/2,
xmin=.52,
xmax=.55,
ymin=0,
ymax=0.03,
xlabel={\small channel parameter $r$},
ylabel={\small achievable rate},
ylabel near ticks,
ticklabel style = {font=\scriptsize},
legend style={font=\scriptsize,},
legend pos=north west,
legend cell align=left,
xtick={.52,.525,.53,.535,.54,.545,.55},
xticklabels={,0.525,0.53,0.535,0.54,0.545,0.55},
title={$\pure(r)/\bec((1-r)^2)$},
]
\addplot[thick] table {QuantumBobPrivateInfo.dat};
\addplot[thick,dashed] table {QuantumBobParityPreprocessing.dat};

\legend{$n=1$,$n=2$}

\end{axis}

\begin{axis}[
    at={(mainplot.north east)},
    xshift=-10pt,
    yshift=-10pt,
    anchor=north east,
    xmin=.54,
    xmax=.546,
    ymin=0,
    ymax=0.002,
    tiny,
    width=\textwidth/4,
    xtick={.54,.541,.542,.543,.544,.545,.546},
    xticklabels={0.54,,0.542,,0.544,,0.546},
    axis background/.style={fill=white},
]
\addplot[thick] table {QuantumBobPrivateInfo.dat};
\addplot[thick,dashed] table {QuantumBobParityPreprocessing.dat};
\fill (axis cs:.543,.0003) circle  (.4mm);

\end{axis}

\end{tikzpicture}
\end{center}
\caption{Plots of $\pinfo{(\cW_{BZ|X}[r])}$ and $\frac12\pinfo(\cW_{BZ|X}[r]^{\otimes 2}\circ \cP_{X^2|X})$ for coding over the $\cW_{BZ|X}[r]=\pure(r)/\bec((1-r)^2)$ channel. Here, $\cP_{X^2|X}$ is the parity-based pre-processing scheme defined in \eqref{eq:parityprep}.  The inset plot shows that $\pinfo{(\cW_{BZ|X}[r])}$ vanishes for a threshold channel parameter value of $r\geq \hat r\approx 0.5424$, while the rate for two channel uses remains positive for $r$ up to roughly $0.545$, thus demonstrating non-additivity for the case of quantum Bob. The dot corresponds to the specific example in the proof of Theorem~\ref{thm:qBob}.}
\label{fig:example-Bob}
\end{figure}

It is interesting to note that noisy pre-processing is necessary for $r$ near $\hat r$, as it can also be shown that the channel is anti-more-capable for $r$ larger than roughly $0.5342$ (specifically, the solution of $h_2(\tfrac{1-r}{2})=r(2-r)$, see Appendix \ref{app:anticapable} for more details). 
On the other hand, neither pre-processing nor regularization are necessary to evaluate the capacity for $r$ less or equal than $\tilde{r}:=\frac{3-\sqrt{5}}{2}\approx 0.3820$, for in this parameter region the channel is degradable. 
Degradability can in principle be determined by searching for a quantum channel which transforms Bob's pure state outputs to Charlie's BEC outputs pointwise, i.e.\ an $\cE_{Z|B}$ such that $\cW_{Z|X}[r](\ketbra x)=\cE_{Z|B}\circ \cW_{B|X}[r](\ketbra x)$ for $x=0,1$. This problem can be cast as a semidefinite program~\cite{heinosaari_extending_2012} in the general case, but here we may appeal to the simpler necessary and sufficient conditions on the existence of such a channel given in Theorem 6 of \cite{heinosaari_extending_2012}. Indeed, since the outputs of $\cW_{B|X}[r]$ are pure, such a channel exists if and only if the fidelity between Charlie's outputs is not smaller than Bob's outputs. The fidelity of Bob's outputs is $r$, while for Charlie the fidelity is just the probability of erasure, $(1-r)^2$. Equating these two gives the threshold value $\tilde r=\tfrac{3-\sqrt 5}2$.
We note that the particular degrading maps is given by the unambiguous state discrimination measurement~\cite{ivanovic_how_1987,dieks_overlap_1988,peres_how_1988}, whose failure probability is precisely the fidelity between the two pure states. 

Fig.~\ref{fig:example-Bob} shows the regularized private information $\frac12\pinfo(\cW_{BZ|X}[r]^{\otimes 2}\circ \mathcal P_{X^2|X})$ of this scheme, versus the private information $\pinfo(\cW_{BZ|X}[r])$ of the bare wiretap channel. As in the proof, by symmetry the rates in each case achieve \eqref{eq:upper_bound}, with the optimal maximal input to the respctive $f$ the uniform distribution. In the language of \cite{ozel_wiretap_2013}, both channels are ``dominantly cyclic shift symmetric''.


\section{Quantum adversarial receiver Charlie}\label{sec:quantum-eve}
Finally, we consider the case that the output to the adversarial receiver Charlie is quantum. 
Again by means of a counterexample, we establish that the private information is generally non-additive for ccq wiretap channels $\cW_{YC|X}$ with classical input system $X$ and classical legitimate receiver $Y$, but quantum adversary $C$. Hence, in the wiretap model, quantum adversaries are strictly different from classical adversaries in the sense that non-additivity effects become possible.


\begin{theorem}\label{thm:qEve}
There exists a ccq wiretap channel $\cW_{YC|X}$ with $\p\left(\cW_{YC|X}\right)>\pinfo\left(\cW_{YC|X}\right)$.
\end{theorem}

\begin{figure}[t]
    \centering
    \begin{tikzpicture}
    \tikzstyle{box}=[draw,minimum width=25mm,minimum height=10mm]
    \node[box] (bpc) at (0,-.8) {$\text{BPC}(1{-}2p)$};
    \node[box] (bsc) at (0,.8) {$\text{BSC}(p)$};
    \node (y) at (3,.8) {$Y$};
    \node (c) at (3,-.8) {$C$};
    \draw (bsc) -- (y);
    \draw (bpc) -- (c);
    \node[inner sep=0,outer sep=0] (ctrl) at (-2,0) {};
    \node (x) at (-4,0) {$X$};
    \draw (x) -- (-2,0);
    \draw (ctrl) |- (bpc);
    \draw (-2,0) |- (bsc);
    \fill (ctrl) circle (.4mm);
    \node[fit=(ctrl) (bsc) (bpc), draw,dashed,inner sep=10pt] (W) {};
    \node[below] at (W.south) {$\cW_{YC|X}[p]$};
    \end{tikzpicture}
    \caption{The channel $\cW_{YC|X}[p]$, composed of $\bsc(p)$ to Bob and $\pure(1-2p)$ to Charlie.}
    \label{fig:QEve_chann}
\end{figure}
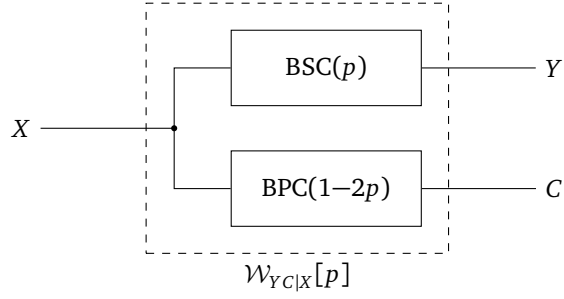

Again we can construct a family of counterexamples, this time based on the wiretap channel $\cW_{YC|X}[p]$ composed of $\bsc(p)$ to Bob and $\pure(1-2p)$ to Charlie, as depicted in Fig.~\ref{fig:QEve_chann}. 
We borrow the block pre-processing based on the repetition code from \cite{smith_structured_2008} to find channels having positive private information under the block pre-processing, but zero for a single use, as in the case of quantum Bob. 
In fact, this example is implicit in \cite{smith_structured_2008}, as $\cW_{YC|X}$ is the result of restricting the quantum channel considered there (the Pauli channel with independent bit and phase errors at identical rates) to standard basis input states.

\begin{proof}
We first show that the private information $\pinfo{\left(\cW_{YC|X}[p]\right)}$ vanishes for $p\geq \hat p$, where $\hat p\approx 0.1241$ satisfies 
\begin{align}\label{eq:quantumEvethreshold}
\frac{(1-2\hat p)^3}{2\hat p(1-\hat p)}=\text{ln}\left(\frac{1-\hat p}{\hat p}\right)\,.
\end{align}
As before, this is established by showing that $f_{\cW[p]}$ is convex for $p\geq \hat p$ by direct calculation of the second derivative. The details of the calculation are given in Appendix \ref{app:B}.

To obtain lower bounds on the private capacity $\p(\cW_{YC|X}[p])$, we employ the block pre-processing from \cite{smith_structured_2008}, adapted to the channel setting.  In particular, consider the pre-processing $\cP_{X^n|X}[q]$ resulting from $n$-bit repetition encoding followed by i.i.d.\ bit-flip noise addition at rate $q$ by the sender. 
Denoting the input to the pre-processing by $X$ and bounding $\pnr(\cW_{YC|X}[p]^{\otimes n}\circ \cP_{X^n|X}[q])$ from \eqref{eq:more-capable} by choosing a uniform $X$, we find via explicit calculation in Appendix \ref{app:C}:
\begin{align}
\p(\cW_{YC|X}[p])&\geq \frac{1}{n}\Big[I(X':Y^{n})_{\omega}-I(X':C^{n})_{\omega}\Big]\\
&=\frac{1}{n}\Big[1-\sum_{s_2^n}p\left(s_2^n\right)H(W|S_2^n=s_2^n)-H\left(\tfrac{1}{2}\rho_{p,q}^{\otimes n}+\tfrac{1}{2}Z^{\otimes n}\rho_{p,q}^{\otimes n}Z^{\otimes n}\right)+nH(\rho_{p,q})\Big]\,.
\end{align} 
Here, $S_2^n \in\{0,1\}^{n-1}$ denotes the syndrome of the repetition code as obtained by the legitimate receiver, $W$ is the value of the logical bit error, and $\rho_{p,q}$ is the output state of $\pure(1-2p)$ for a zero-valued bit in the code-block that has undergone the pre-processing bit-flip channel. This is precisely Eq.~(2) in \cite{smith_structured_2008}. As numerically evaluated therein, the expression remains positive at least up to the threshold $\bar{p}\approx0.129$, which is obtained from $n=400$ and $q=0.32$. 
\end{proof}

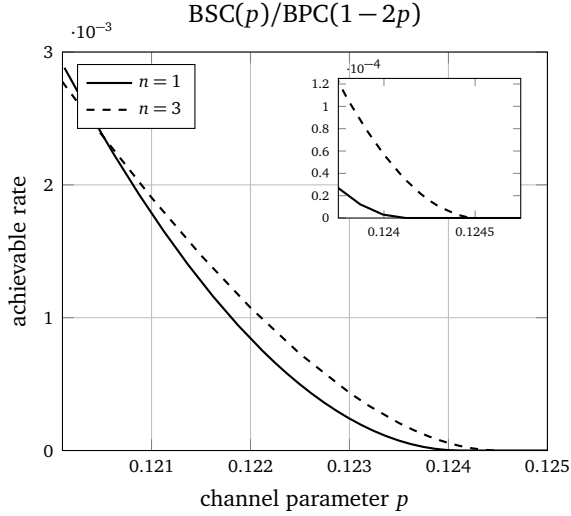
\begin{figure}[ht]
\begin{center}
\begin{tikzpicture}

\begin{axis}[
name=mainplot,
grid,
enlargelimits = false,
width=\textwidth/2,
xmin=.1201,
xmax=.125,
ymin=0,
ymax=0.003,
xlabel={\small channel parameter $p$},
ylabel={\small achievable rate},
ylabel near ticks,
ticklabel style = {font=\scriptsize},
legend style={font=\scriptsize,},
legend pos=north west,
legend cell align=left,
xtick={.121,.122,.123,.124,.125},
xticklabels={0.121,0.122,0.123,0.124,0.125},
title={$\bsc(p)/\pure(1-2p)$}
]
\addplot[thick] table {QuantumEvePrivateInfo.dat};
\addplot[thick,dashed] table[x index=1, y index=3] {QuantumEveRepNoiseProcessing.dat};

\legend{$n=1$,$n=3$}

\end{axis}

\begin{axis}[
    at={(mainplot.north east)},
    xshift=-10pt,
    yshift=-10pt,
    anchor=north east,
    xmin=.12375,
    xmax=.12475,
    ymin=0,
    ymax=0.000125,
    tiny,
    width=\textwidth/4,
    xtick={.124,.1245},
    xticklabels={0.124,0.1245},
    axis background/.style={fill=white},
]
\addplot[thick] table {QuantumEvePrivateInfo.dat};
\addplot[thick,dashed] table[x index=1, y index=3] {QuantumEveRepNoiseProcessing.dat};

\end{axis}

\end{tikzpicture}
\end{center}
\caption{Plots of $\pinfo{(\cW_{YC|X}[p])}$ and $\max_{q}\tfrac13 \pinfo(\cW_{YC|X}[p]^{\otimes 3}\circ \cP_{X^3|X}[q])$ for coding over the $\cW_{YC|X}[p]=\bsc(p)/\pure(1-2p)$ channel. We refer to the main text for the definition of the pre-processing map $\cP_{X^3|X}[q]$. The inset plot shows that the private information is zero beyond $p=\hat p\approx 0.1241$, while the achievable rate for three uses of the channel remains positive up to $p\approx 0.1245$. This demonstrates non-additivity for the case of quantum Eve.}
\label{fig:example-Eve}
\end{figure}

In the context of key distillation, the threshold value $\hat p$ was found by numerical optimization in \cite[Footnote 21]{smith_structured_2008}. 
Furthermore, the wiretap channel $\cW_{YC|X}[p]$ becomes anti-degradable for $p$ larger or equal to
$\tilde{p}:=\frac{2-\sqrt2}{4}\approx0.1464$, 
at which point we have $\p(\cW_{YC|X}[\tilde{p}])=0$.
Again we make use of \cite[Theorem 6]{heinosaari_extending_2012} to establish anti-degradability. 
Here Charlie's states are pure, with a fidelity $1-2p$, while Bob's states have a fidelity of $2\sqrt{p(1-p)}$. Equating these two gives the threshold $\tilde p=\tfrac{2-\sqrt 2}4$. 
The degrading map in this case is simply the Helstrom measurement to distinguish the pure states~\cite{helstrom_quantum_1976}.
This value is identical to thresholds found in the upper bounds for key distillation~\cite{fuchs_optimal_1997,kraus_lower_2005,moroder_one-way_2006} or private communication over the corresponding quantum channel with independent bit and phase flip errors~\cite{smith_private_2008,smith_additive_2008}. 
Note that it remains an open question if the private capacity is non-zero all the way up to the degradability threshold $\tilde{p}$.

In Fig.~\ref{fig:example-Eve} we provide a comparison plot between $\pinfo{(\cW_{YC|X}[p])}$ and $\max_{q}\tfrac13 \pinfo(\cW_{YC|X}[p]^{\otimes 3}\circ \cP_{X^3|X}[q])$ for the $n=3$ pre-processing scheme. We find a threshold of $0.1245$, and perhaps coincidentally the optimal $q$ also appears to be $0.32$ in this case.





\section{Discussion}\label{sec:conclusion}

We determined for which cases the private information of hybrid classical-quantum wiretap channels is non-additive. We found additivity violations when either of the two receivers becomes quantum; interestingly without any entanglement being present. On the other hand, we also showed that for quantum inputs but classical receivers the private information remains additive. That is, entangled input states are of no help. We note that the setting in \cite{smith_structured_2008} is already an instance of non-additivity without entangled inputs, because the combination of repetition coding and standard basis inputs produces separable states. This is precisely what we use in Section \ref{sec:quantum-eve}.  Moreover, we can regard the parity encoding in Section \ref{sec:quantum-bob} as a phase error-detecting code with stabilizer $XX$, since this operator also stabilizes the outputs of \eqref{eq:parityprep}. However, the link between private and quantum coding does not hold for more general preprocessing. For instance, \cite{fern_lower_2008} studies the effects of preprocessing using the five qubit code on various Pauli channels. But this cannot be interpreted as a classical preprocessing for a classical wiretap channel, since in the five qubit code the relative phases of the codewords play a decisive role, yet they disappear in any classical encoding. 

Similar to general additivity questions in quantum Shannon theory, it remains open to quantify the magnitude of how non-additive the private information can become. The results here may shed some light on the role of degenerate codes in non-additivity. Interestingly, pre-processing based on repetition coding does not lead to non-additivity for the quantum Bob example,  nor does parity encoding lead to non-additivity for the quantum Charlie example. Is this a general trend or just a coincidence? Another interesting question to resolve is if more capable and less noisy wiretap channels have an additive private information\,---\,as they do in the classical case. More broadly, we might ask how far we can push the question about the quantumness needed to witness non-additivity phenomena in Shannon information theory. A natural candidate that remains open is to resolve if Marton's inner bound for general broadcast channels \cite{marton_coding_1979} is additive or not \cite{anantharam_evaluation_2019}. Insights from quantum information theory as presented here might be able to shine some light on this long-standing question.



\paragraph*{Acknowledgements} We thank Andreas Winter for discussions related to the topic of this paper. JMR was supported by the Swiss National Science Foundation (SNSF) via the National Center of Competence in Research QSIT, as well as the Air Force Office of Scientific Research (AFOSR) via grant FA9550-19-1-0202. 

\printbibliography[heading=bibintoc,title=References]

\appendix

\section{Private information in the case of quantum Bob}
\label{app:A}
\subsection{Anti-less-noisy}
\noindent Here we calculate the second derivative of $f_{\cW[r]}$. 
More precisely, let $\omega_{X'BZ}=\sum_{x}p_{X}(x)\dyad{x}{x}_{X'}\otimes \cW_{BZ|X}[r]\left(\dyad{x}{x}_{X}\right)$ and $\cW_{BZ|X}[r]:=\pure(r)/\bec{\left((1-r)^2\right)}$ with $q=p_X(0)$. 
Then $f_{\cW[r]}(\rho_X)=I(X':B)_\omega-I(X':E)_\omega$ is a function of $q$, and we are interested in $\frac{\partial^2}{{\partial q}^2}\left[f_{\cW[r]}(\rho_X)\right]$.
Since the output to Bob is a pure state, in this case we have 
\begin{align}
\label{appA:f}
    I(X':B)_{\omega}-I(X':Z)_{\omega}= H(\omega_{B})-H(\omega_{Z})+h_{2}\left((1-r)^2\right)\,.
\end{align}
Thus, 
we have
$    \frac{\partial^2}{{\partial q}^2}\left[f_{\cW[r]}(\rho_X)\right]
    = \frac{\partial^2}{{\partial q}^2}\left[ H(\omega_{B}) \right]
    -\frac{\partial^2}{{\partial q}^2}\left[ H(\omega_{Z})\right]
    $.
To evaluate this expression, we will make use of the following lemma.

\begin{lemma} 
[Second derivative of Shannon entropy] 
\label{lem:entropy}
\begin{align}
    \frac{\partial^2}{{\partial q}^{2}}\left[\sum_{i} -\lambda_{i}(q)\log(\lambda_{i}(q))\right]=-\sum_{i}\left[\frac{\partial^2\lambda_{i}(q)}{{\partial q}^2}\left[\frac{1}{\ln(2)}+\log(\lambda_{i}(q))\right]+\frac{1}{\ln(2)}\frac{1}{\lambda_{i}(q)}\left(\frac{\partial\lambda_{i}(q)}{\partial q}\right)^2\right]\,.
\end{align}
\end{lemma}
\noindent 
The channel to Charlie is just $\bec((1-r)^2)$, whose output $\omega_C$ is diagonal with probabilities $\lambda_0=q(1-(1-r)^2)$, $\lambda_1=(1-q)(1-(1-r)^2)$, and $(1-r)^2$. 
Using Lemma \ref{lem:entropy} we thus have
\begin{align}
    \frac{\partial^2}{{\partial q}^2}\left[H(\omega_{Z})\right]=- \sum_{i=0}^{1} \left[\frac{1}{\ln(2)}\frac{1}{\lambda_{i}} \left(\frac{\partial \lambda_{i}}{\partial q} \right)^{2}\right]
    = -\frac{(1-(1-r)^2)}{\ln(2)}\left[\frac{1}{q} +\frac{1}{(1-q)}
    \right]\,.
\end{align}

Meanwhile, the channel to Bob is $\pure(r)$, yielding the output state $\omega_B=q\ketbra{\psi^0}_B+(1-q)\ketbra{\psi^1}_B$, where $\ket{\psi^{0}}=\ket{0}$ and $\ket{\psi^{1}}=r \ket{0}+\sqrt{1-r^2}\ket{1}$.
Hence 
\begin{align}
\omega_{B}=\left[q+(1-q)r^2\right]\dyad{0}{0}_{B} + (1-q)(1-r^2)\dyad{1}{1}_{B} + (1-q)r\sqrt{1-r^2}\left[\dyad{0}{1}_{B} + \dyad{1}{0}_{B}\right]\,.
\end{align}
Via the characteristic polynomial, the eigenvalues are found to be $\lambda_{\pm}=\frac{1}{2}(1\pm \tilde{g}(r,q))$ with $\tilde{g}(r,q)= \sqrt{1-4(1-r^2)(q-q^2)}$.
It thus follows that 
\begin{align}
\frac{\partial}{{\partial q}}\lambda_{\pm}=\mp\frac{1}{2}\,\frac{2(1-r^2)(1-2q)}{\tilde{g}(r,q)}\quad\mathrm{and}\quad\frac{\partial^2}{{\partial q}^2} \lambda_{\pm}= \pm\frac{1}{2}\,\frac{4(1-r^2){\tilde{g}(r,q)}^2-(2(1-r^2)(1-2q))^2}{\tilde{g}(r,q)^3}\,.
\end{align}
Using Lemma \ref{lem:entropy}, we obtain 
\begin{align}\label{eq:qBob_derivative}
&\frac{\partial^2}{{\partial q}^2}\left[H(\omega_{B})\right]-\frac{\partial^2}{{\partial q}^2}\left[H(\omega_{Z})\right]\notag\\
&=-\left[  \frac{\partial^2 \lambda_+ }{{\partial q}^2}\left[\frac{1}{\ln(2)}+\log(\lambda_+)\right]+\frac{1}{\ln(2)}\frac{1}{\lambda_+}\left(\frac{\partial\lambda_{+} }{\partial q}\right)^2  \right]-\left[  \frac{\partial^2 \lambda_- }{{\partial q}^2}\left[\frac{1}{\ln(2)}+\log(\lambda_-)\right]+\frac{1}{\ln(2)}\frac{1}{\lambda_-}\left(\frac{\partial\lambda_{-} }{\partial q}\right)^2  \right]\notag\\
&\quad+ \frac{(1-(1-r)^2)}{\ln(2)}\left[\frac{1}{q} +\frac{1}{(1-q)}\right]\,.
\end{align}
We may then rewrite the second derivative in terms of $\tilde{g}(r,q)$ as
\begin{align}
\frac1{\tilde{g}(r,q)^2}\left(
4 (1 - r)^3 (1 + r) + \frac{2r-1}{(1 - q) q}-\frac{4 r^2 (1 - r^2) \tanh^{-1}(\tilde{g}(r,q))}{\tilde{g}(r,q)}\right)\,.
\label{eq:qbobderivsimp}
  \end{align}
This expression is symmetric in $q$ around $q=\frac{1}{2}$.
In order to find the threshold parameter $\hat{r}$, we set $q=1/2$, while setting \eqref{eq:qbobderivsimp} to zero. 
Rearranging yields \eqref{eq:quantumBobthreshold}. 
One can then numerically verify that the expression is positive away from $q=\frac{1}{2}$, thus proving convexity of $f_{\cW[r]}(\rho_X)$ for $r\geq \hat{r}$.\qed

\subsection{Anti-more-capable}
\label{app:anticapable}
The channel is anti-more-capable when $f_{\cW[r]}(\rho_X)\leq 0$, i.e.\ when the naive rate is zero. Again by symmetry in $p_X$, the uniform distribution is decisive. 
Setting $p_X=1/2$ gives $H(\omega_B)=h_2(\tfrac12(1-r))$, $H(\omega_Z)=-(1-r^2)\log (1-r^2)-(1-(1-r^2))\log \tfrac12(1-(1-r^2))$, and therefore the rate expression in \eqref{appA:f} is just $h_2(\tfrac12(1-r))-(2-r)r$. 

\section{Private information in the case of quantum Charlie}
\label{app:B}

Here we calculate the second derivative of $f_{\cW[p]}$. 
Let $\omega_{X'YC}=\sum_{x}p_{X}(x)\dyad{x}{x}_{X'}\otimes \cW_{YC|X}[p]\left(\dyad{x}{x}_{X}\right)$ with $\cW_{YC|X}[p]:=\bsc(p)/\pure(1-2p)$ and $p_X(0)=q$. 
Again $f_{\cW[p]}(\rho_X)=I(X':Y)_{\omega}-I(X':C)_{\omega}$ is a function of $q$, and we are interested in $\frac{\partial^2}{{\partial q}^2}\left[f_{\cW[p]}(\rho_X)\right]$.
Now the output to Charlie is pure, so we have 
\begin{align}
 I(X':Y)_{\omega}-I(X':C)_{\omega}= H(\omega_{Y})-H(\omega_{C})-h_{2}(p)
 \end{align}
and hence $
    \frac{\partial^2}{{\partial q}^2}\left[f_{\cW[p]}(\rho_X)\right]
    = \frac{\partial^2}{{\partial q}^2}\left[ H(\omega_{Y}) \right]
    -\frac{\partial^2}{{\partial q}^2}\left[ H(\omega_{C})\right]
    $.
The output to Bob is classical, with probabilities $\lambda_0=q(1-p)+(1-q)p$ and $\lambda_1=(1-q)(1-p)+qp$. 
Appealing to Lemma \ref{lem:entropy} gives
\begin{align}
\frac{\partial^2}{{\partial q}^2}\left[H(\omega_{Y})\right]=- \sum_{i=0}^{1} \left[\frac{1}{\ln(2)}\frac{1}{\lambda_{i}} \left(\frac{\partial \lambda_{i}}{\partial q} \right)^{2}\right]= -\frac{(1-2p)^2}{\ln(2)}\left[\frac{1}{q(1-p)+(1-q)p} +\frac{1}{qp+(1-q)(1-p)}\right]\,.
\end{align}

The channel to Charlie is $\pure(1-2p)$, yielding the output state
$\omega_{C}=\sum_{x} p_{X}(x)\dyad{\phi^{x}}{\phi^{x}}_{C}$,
where $\ket{\phi^x}=\sqrt{1-p}\ket{0}+(-1)^x\sqrt p \ket{1}$.
In terms of $q$, the state is  
\begin{align}
\omega_{C}=(1-p) \dyad{0}{0}_{C}+ p \dyad{1}{1}_{C}+ (2q-1)\sqrt{p(1-p)} [\dyad{0}{1}_{C} + \dyad{1}{0}_{C}]\,.
\end{align}
Via the characteristic polynomial, the eigenvalues are found to be $\lambda_{\pm}=\frac{1}{2}(1\pm g(p,q))$, with $g(p,q)=\sqrt{1-4p(1-p)[1-(2q-1)^2]}$.
It then follows that
$\frac{\partial}{{\partial q}}\lambda_{\pm}=\pm\frac{1}{2}\,\frac{8\,p(1-p)(2q-1)}{g(p,q)}$ and
\begin{align}
\frac{\partial^2}{{\partial q}^2} \lambda_{\pm}= \pm\frac{1}{2}\, \frac{16\,p(1-p)g(p,q)^2-(8p(1-p)(2q-1))^2}{g(p,q)^3}\,.
\end{align}
Using Lemma \ref{lem:entropy}, we obtain for the second derivative
\begin{align}
&\frac{\partial^2}{{\partial q}^2}\left[H(\omega_{Y})\right]-\frac{\partial^2}{{\partial q}^2}\left[H(\omega_{C})\right]\notag\\
&=\left[  \frac{\partial^2 \lambda_+ }{{\partial q}^2}\left[\frac{1}{\ln(2)}+\log(\lambda_+)\right]+\frac{1}{\ln(2)}\frac{1}{\lambda_+}\left(\frac{\partial\lambda_{+} }{\partial q}\right)^2  \right]+\left[  \frac{\partial^2 \lambda_- }{{\partial q}^2}\left[\frac{1}{\ln(2)}+\log(\lambda_-)\right]+\frac{1}{\ln(2)}\frac{1}{\lambda_-}\left(\frac{\partial\lambda_{-} }{\partial q}\right)^2  \right]\notag\\
&\quad-\frac{(1-2p)^2}{\ln(2)}\left[\frac{1}{q(1-p)+p(1-q)}+\frac{1}{qp+(1-q)(1-p)}\right]\,.
\end{align}
We may then rewrite the second derivative in terms of $c:=g^2(p,q)=1 - 16 p (1-p) q(1-q)$ as
\begin{align}
-\frac1{\log 2}\cdot\left(
\frac{(1 - 2 p)^2}{(q + p (1 - 2 q)) (1 - q - p (1 - 2 q))}
- \frac{4 (1 - p) p (1 - 2 q)^2}{c(1 - q)q} 
+  \frac{8 (1 - 2 p)^2 (1 - p) p}{c^{3/2}} \log\frac{1 - \sqrt{c}}{1 + \sqrt{c}}
    \right)\,.
\end{align}
Again, this expression is symmetric in $q$ around $q=\frac{1}{2}$. In order to find the threshold parameter $\hat{p}$, we set $q=\frac{1}{2}$, for which the middle term vanishes, leaving just 
\begin{align}
\frac4{\log 2}\cdot\frac{2p(1-p)\log\frac{1-p}{p}-(1-2p)^3}{1-2p}\,.
\end{align}
Setting this to zero and restricting $p\in [0,\frac{1}{2}]$ gives \eqref{eq:quantumEvethreshold}. One can then numerically verify that the second derivative is positive away from $q=\frac{1}{2}$, thus proving convexity of $f_{\cW[p]}(\rho_X)$ for $p\geq\hat{p}$.\qed

\section{Lower bound on $\p(\cW[p])$}\label{app:C}

\noindent We evaluate 
a lower bound to $\pnr(\cW_{YC|X}[p]^{\otimes n}\circ \cP_{X^n|X}[q])$ by evaluating $I(X':Y^n)_{\omega}-I(X':C^n)_{\omega}$ for a uniform binary input distribution (to simplify the calculation), such that
\begin{align}
\omega_{X'Y^{n}C^{n}}=\frac{1}{2} \sum_{x\in \{0,1\}} \dyad{x}{x}_{X'}\otimes \cW_{YC|X}[p]^{\otimes n}\circ \cP_{X^n|X}[q]\left(\dyad{x}{x}_{X}\right)
\end{align}
with $\cW_{YC|X}[p]$  the $\bsc{(p)}/\pure{(1-2p)}$ channel and $\cP_{X^n|X}[q]$  the pre-processing map resulting from the $n$-bit repetition encoding followed by i.i.d. bit-flip noise addition at rate $q$ by the sender. 
Here, $X'$ denotes a copy of the classical input system. Applying the pre-processing map we obtain
\begin{align}
\omega_{X'Y^{n}C^{n}}
=\frac{1}{2} \sum_{x\in \{0,1\}} \dyad{x}{x}_{X'}\otimes \cW_{YC|X}[p]^{\otimes n}\left(\sum_{w^n\in \{0,1\}^{n}} q^{|w^n\oplus (x)^{\times n}|}(1-q)^{n-|w^n\oplus (x)^{\times n}|} \dyad{w^n}{w^n}_{X^n} \right)\,,
\end{align}
where $(x)^{\times n}$ is either an all-zeros or all-ones bit-string, depending on the value of $x$ and $w^n$ denotes the bit-error pattern on the all-zeros string, and $X^n$ denotes the code-block. We can evaluate $I(X':C^n)_{\omega}$ by considering the reduced state
$\omega_{X'C^n}$, where
\begin{align}
\omega_{X'C^n}=\frac{1}{2} \sum_{x\in \{0,1\}} \dyad{x}{x}_{X'}\otimes {\cW_{C|X}[p]}^{\otimes n}\left(\sum_{w^n\in \{0,1\}^{n}} q^{|w^{n}\oplus (x)^{\times n}|}(1-q)^{n-|w^{n}\oplus (x)^{\times n}|} \dyad{w^n}{w^n}_{X^n} \right)
\end{align}
for $\cW_{C|X}[p]$ the $\pure{(1-2p)}$-channel, such that $\cW_{C|X}[p]\left(\dyad{w}{w}_{X}\right)=Z^{w}\dyad{\phi^{0}}{\phi^{0}}_{C} Z^{w}$ for $w \in \{0,1\}$ and Pauli matrix $Z$, where
\noindent $\ket{\phi^{0}}= \sqrt{1-p}\ket{0} + \sqrt{p} \ket{1}$ with $p$ being the channel parameter of the $\pure{(1-2p)}$.
We thus have
\begin{align}
\omega_{X'C^{n}}=\frac{1}{2} \sum_{x\in \{0,1\}} \dyad{x}{x}_{X'}\otimes \sum_{w^n\in \{0,1\}^{n}} q^{|w^n\oplus (x)^{\times n}|}(1-q)^{n-|w^n\oplus (x)^{\times n}|} \left(\bigotimes _{i=1}^n Z_{C_i}^{w_i} \dyad{\phi^{0}}{\phi^{0}}_{C_i} Z_{C_i}^{w_i}\right)
\end{align}
such that Charlie obtains the state
\begin{align}
\omega_{C^n}=\frac{1}{2}\sum_{w^n\in \{0,1\}^{n}} q^{|w^n|}(1-q)^{n-|w^n|}+q^{n-|w^n|}(1-q)^{|w^n|} \left(\bigotimes _{i=1}^n Z_{C_i}^{w_i} \dyad{\phi^{0}}{\phi^{0}}_{C_i}Z_{C_i}^{w_i}\right)\,.
\end{align}
\noindent Since $\omega_{X'C^n}$ is a cq-state, we can evaluate its mutual information as
\begin{align}
I(X':C^n)_{\omega}=H\left(\omega_{C^n}\right)-\frac{1}{2} \sum_{x\in \{0,1\}}H\left(\omega_{C^n}^{x}\right)
\end{align}
where $\omega_{C^n}^{x}$ is defined as
\begin{align}
\omega_{C^n}^{x}=\sum_{w^n\in \{0,1\}^{n}} q^{|w^n\oplus (x)^{\times n}|}(1-q)^{n-|w^n\oplus (x)^{\times n}|} \left(\bigotimes _{i=1}^n Z_{C_i}^{w_i} \dyad{\phi^{0}}{\phi^{0}}_{C_i} Z_{C_i}^{w_i}\right)
\end{align}
such that $\omega_{X'C^{n}}=\frac{1}{2} \sum_{x\in \{0,1\}} \dyad{x}{x}_{X}\otimes \omega_{C^n}^{x}$. Notice that $\omega_{C^n}^{x\oplus 1}= Z^{\otimes n}\omega_{C^n}^{x}Z^{\otimes n}$ and we can thus rewrite $\omega_{C^n}$ as $\omega_{C^n}=\frac{1}{2} \left(\omega_{C^n}^{0}+Z^{\otimes n}\omega_{C^n}^{0}Z^{\otimes n}\right)$. Evaluating $I(X':C^n)_{\omega}$ we then have
\begin{align}
I(X':C^n)_{\omega}&=H\left(\tfrac{1}{2} \omega_{C^n}^{0}+ \tfrac{1}{2} Z^{\otimes n}\omega_{C^n}^{0}Z^{\otimes n}\right)-\tfrac{1}{2}\left\{H\left(\omega_{C^n}^{0}\right)+H\left(Z^{\otimes n}\omega_{C^n}^{0}Z^{\otimes n}\right)\right\}\\
&=H\left(\tfrac{1}{2} \omega_{C^n}^{0}+ \tfrac{1}{2} Z^{\otimes n}\omega_{C^n}^{0}Z^{\otimes n}\right)-H\left(\omega_{C^n}^{0}\right)\,,
\end{align}
since $H\left(\omega_{C^n}^{x\oplus 1}\right)=H\left(\omega_{C^n}^{x}\right)$ due to isometric invariance of the von Neumann entropy.
Our definition of $\omega_{C^n}^{0}$ coincides with the definition of $\rho_{p,q}^{\otimes m}:= (1-q)\ketbra{\phi_{+}}+ q \ketbra{\phi_{-}}$ with $\ket{\phi_{\pm}}= \sqrt{1-p}\ket{0}\pm \sqrt{1-p}\ket{1} $ from \cite{smith_structured_2008}, so that we can write
\begin{align}
I(X':C^n)_{\omega}=H\left(\tfrac{1}{2} \rho_{p,q}^{\otimes n}+ \tfrac{1}{2} Z^{\otimes n}\rho_{p,q}^{\otimes n}Z^{\otimes n}\right)-H\left(\rho_{p,q}^{\otimes n}\right)=H\left(\tfrac{1}{2} \rho_{p,q}^{\otimes n}+ \tfrac{1}{2} Z^{\otimes n}\rho_{p,q}^{\otimes n}Z^{\otimes n}\right)-n\,H\left(\rho_{p,q}\right)\,,
\end{align}
since the von Neumann entropy is additive for product states.
We now proceed to evaluate $I(X':Y^n)_{\omega}$. To do this, we look at the combination of the $\bsc{(p)}$ with the pre-processing map $\cP_{X^n|X}[q]$. We can thus write the state $\omega_{X'Y^{n}}$ as 
\begin{align}
\omega_{X'Y^n}&=\tfrac{1}{2}\sum_{x\in\{0,1\}}\dyad{x}{x}_{X'}\otimes \cW_{Y|X}[p]^{\otimes n}\circ \cP_{X^n|X}[q]\left(\dyad{x}{x}_{X}\right)\\
&=\tfrac{1}{2}\sum_{x\in\{0,1\}}\dyad{x}{x}_{X}\otimes \sum_{w^n\in\{0,1\}^{n}} \dyad{(x)^{\times n} \oplus w^n}{(x)^{\times n} \oplus w^n}_{Y^n} 
(1-\tilde{p})^{n-|w^{n}|}\tilde{p}^{|w^{n}|}\,,
\end{align}
where $\cW_{Y|X}[p]^{\otimes n}\circ \cP_{X^n|X}[q]$ is the combined bit-flip channel with bit-flip probability $\tilde{p}=q(1-p)+p(1-q)$, and $p$ and $q$ are the bit-flip probabilities of the respective channels $\bsc{(p)}$ and pre-processing channels $\cP_{X^n|X}[q]$, and $w^n$ is the bit-flip error pattern occurring on the encoded state $\dyad{(x)^{\times n}}{(x)^{\times n}}$. Equivalently, due to isometric invariance of the von Neumann entropy, one may first apply the unitary decoding map $\identity_{X'}\otimes\decoder$ of the repetition code to the state $\omega_{X'Y^n}$ and then evaluate $I(X':V_{1}S_{2}^{n})_{\sigma}$ on the resultant state
\begin{align}
\sigma_{X'V_{1}S_{2}^{n}}= \left(\identity_{X'}\otimes\decoder\right)\left(\omega_{X'Y^n}\right)\,.
\end{align}
Here, $V_{1}$ is the first bit of the $n$-bit repetition code block, and $S_{2}^{n}:=(S_2,...,S_n)$ denotes the set of $n-1$ syndrome bits. Implementing the decoding map $\decoder$ by performing a CNOT gate from $Y_1$ to each of the bits of $Y_{2}^{n}=(Y_2,..., Y_{n})$ sequentially allows us to write
\begin{align}
\sigma_{X'V_1S_{2}^{n}}=&\;\tfrac{1}{2}\sum_{x\in\{0,1\}}\dyad{x}{x}_{X'}\otimes \sum_{w_{1}\in\{0,1\}}\sum_{w_{2}^{n}\in \{0,1\}^{n-1}} \dyad{x\oplus w_{1}}{x \oplus w_{1}}_{V_{1}}\otimes \dyad{w_{2}^{n}\oplus (w_1)^{\times (n-1)}}{w_{2}^{n}\oplus (w_1)^{\times (n-1)}}_{S_{2}^{n}}\notag\\
&(1-\tilde{p})^{m-\left(|w_1|+|w_{2}^{m}|\right)}\tilde{p}^{|w_1|+|w_{2}^{m}|}\,,
\end{align}
where $(w_1)^{\times (n-1)}$ is a bit-string of length $n-1$ that is either all-zeros or all-ones depending on the value of $w_1$. We now state two observations about the marginals of this state:
\begin{fact}[$\sigma_{V_{1}S_{2}^n}$ is a product state]
One may write the state $\sigma_{\codeblock}$ as
\begin{align}
\sigma_{\codeblock}&=\underbrace{\tfrac{1}{2} \left(\dyad{0}{0}_{V_{1}}+\dyad{1}{1}_{V_{1}}\right)}_{ \sigma_{{V}_{1}}}\otimes \underbrace{\tfrac{1}{2} \sum_{\strings} (1-\tilde{p})^{n-|w_{2}^{n}|} \tilde{p}^{|w_{2}^{n}|}\left[\dyad{w_{2}^{n}}{w_{2}^{n}}_{S_{2}^{n}}+ \dyad{w_{2}^{n}\oplus 1}{w_{2}^{n}\oplus 1}_{S_{2}^{n}} \frac{\tilde{p}}{(1-\tilde{p})}\right]}_{\sigma_{S_{2}^{n}}}\\
&=\sigma_{{V}_{1}} \otimes \sigma_{S_{2}^{n}}\,,
\end{align}
where $\sigma_{{V}_{1}}$ is the maximally mixed state on system $V_1$, such that $H(V_1)=1$. 
\end{fact}

\begin{fact}[$\sigma_{X'S_{2}^n}$ is a product state]
One may write the state $\sigma_{X'S_{2}^n}$ as 
\begin{align}
\sigma_{X'S_{2}^n}&= \underbrace{ \left(\tfrac{1}{2} \sum_{x\in\{0,1\} }\dyad{x}{x}_{X'}\right)}_{ \sigma_{X'}}\otimes \underbrace{\left(\tfrac{1}{2}\sum_{{w_1}\in \{0,1\}}\sum_{\strings} (1-\tilde{p})^{n-|w^n|}\tilde{p}^{w} \dyad{w_{2}^{n}\oplus \left(w_{1}\right)^{\times (n-1)}}{w_{2}^{n}\oplus \left(w_{1}\right)^{\times (n-1)}}\right)}_{\sigma_{S_{2}^{n}}}\\
&= \sigma_{X'}\otimes \sigma_{S_{2}^{n}}\,.
\end{align}
\end{fact}

We can thus evaluate $I(X':Y^n)_{\omega}$ as
\begin{align}
& I(X':Y^n)_{\omega}\notag\\
&=I(X':V_{1}S_{2}^{n})_{\sigma}\quad\quad\quad\quad\quad\quad\quad \text{[via isometric invariance of the von Neumann entropy]}\\
& = H(\codeblock)_{\sigma}-H(\codeblock|X')_{\sigma} \quad\quad\quad\quad\quad\quad\quad\ \text{[via definition]}\\
& = H(V_1)_{\sigma}+H(\syndrome)_{\sigma}-H(\codeblock|X')_{\sigma} \quad\quad\quad\quad \text{[via $\sigma_{\codeblock}=\sigma_{V_1}\otimes \sigma_{S_{2}^{n}}$ + chain rule]}\\
&= 1+ H(\syndrome)_{\sigma}-H(\codeblock|X)_{\sigma}\quad\quad\quad\quad\quad\quad\quad \text{[ since $\sigma_{V_1}$ is maximally mixed]}\\
&= 1+ H(\syndrome)_{\sigma}-H(\syndrome|X')_{\sigma}-H(V_1|S_{2}^{n},X')_{\sigma}\quad\text{[via chain rule]}\\
&= 1-H(V_1|S_2^n,X')_{\sigma}\quad\quad\quad\quad\quad\quad\quad \quad\quad\quad\quad\quad\text{[via $\sigma_{X' S_{2}^{n}}=\sigma_{X'}\otimes \sigma_{S_{2}^{n}}$]}\\
&= 1-H(W_1\oplus X|S_{2}^{n},X)_{\sigma}\quad\quad\quad\quad\quad\quad\quad\quad\quad\quad\text{[since $V_{1}=W_{1}\oplus X$ and $X'=X$ via definition]}\\
&= 1-H(W_1|S_{2}^{n})_{\sigma}\\
&= 1-\sum_{s_{2}^{m}\in \{0,1\}^{n-1}} p(s_{2}^{n}) H(W_1|S_{2}^{n}=s_{2}^{n})\,.
\end{align}
We thus have 
$
I(X':Y^{n})_{\omega}= 1-\sum_{s_{2}^{n}\in \{0,1\}^{n-1}} p(s_{2}^{n}) H(W_1|S_{2}^{n}=s_{2}^{n})
$, 
such that 
\begin{align}
I(X':Y^n)_{\omega}-I(X':C^n)_{\omega}= 1-\sum_{s_{2}^{n}\in \{0,1\}^{n-1}} p(s_{2}^{n}) H(W_1|S_{2}^{n}=s_{2}^{n})-H\left(\tfrac{1}{2} \rho_{p,q}^{\otimes n}+ \tfrac{1}{2} Z^{\otimes n}\rho_{p,q}^{\otimes n}Z^{\otimes n}\right)+n\,H\left(\rho_{p,q}\right)\,.
\end{align}


\end{document}